\documentclass[letterpaper, 10 pt, conference]{ieeeconf}  

\usepackage{adjustbox} 
\IEEEoverridecommandlockouts                              

\usepackage{times}
\usepackage[utf8]{inputenc}		
\usepackage{graphicx}			
\usepackage{graphics}
\usepackage{psfrag}
\usepackage{epic,eepic}
\usepackage{amsmath}			
\usepackage{amssymb}
\usepackage{color}
\usepackage{tikz}
\usetikzlibrary{arrows}
\usetikzlibrary{positioning}
\usepackage{epstopdf}
\usepackage{enumerate}
\usepackage{ragged2e}
\usepackage{subfigure}

\usepackage{tikz}
\usetikzlibrary{positioning,calc}
\tikzstyle{block} = [draw, fill=blue!20, rectangle, minimum height=3em, minimum width=6em]
\tikzstyle{sum} = [draw, fill=blue!20, circle, node distance=1cm]
\tikzstyle{input} = [coordinate]
\tikzstyle{output} = [coordinate]
\tikzstyle{pinstyle} = [pin edge={to-,thin,black}]
\usepackage{tikzscale}
\newcommand{\bbr}{\mbox{$\mathbb{R}$}}

\newcommand{\bbz}{\mbox{$\mathbb{Z}$}}

\newcommand{\mcx}{\mbox{$\mathcal X$}}
\newcommand{\mcu}{\mbox{$\mathcal U$}}
\newcommand{\mcr}{\mbox{$\mathcal R$}}

\newcommand{\mcw}{\mbox{$\mathcal W$}}
\newcommand{\mcd}{\mbox{$\mathcal D$}}

\newcommand{\diag}{\mbox{$\mathrm{diag}$}}

\newtheorem{thm}{\noindent{\bf Theorem}}
\newtheorem{lem}{\noindent{\bf Lemma}}
\newtheorem{dfn}{\noindent{\bf Definition}}
\newtheorem{rmk}{\noindent{\bf Remark}}

\title{\LARGE \bf Combining Robust Control and Machine Learning for Uncertain Nonlinear Systems Subject to Persistent Disturbances*}

\author{A. Banderchuk$^{1}$, D. Coutinho$^{2}$ and E. Camponogara$^{2}$
\thanks{© 2023 IEEE. Personal use of this material is permitted. Permission from IEEE must be obtained for all other uses, in any current or future media, including reprinting/republishing this material for advertising or promotional purposes, creating new collective works, for resale or redistribution to servers or lists, or reuse of any copyrighted component of this work in other works.
}
\thanks{*This work was partially supported by CAPES under grants 88887.153840/2017-00 (SIU) and 88887.522001/2020-00 (C-CAIT, SticAmSud), and CNPq under grants 163116/2021-0,  303289/2022-8/PQ and 308624/2021-1/PQ. }
\thanks{$^{1}$A. Banderchuk is with Post-graduate Program in Automation and Systems Engineering, Universidade Federal de Santa Catarina, Florianópolis, SC, 88040-900, Brazil, e-mail: {\tt\small ana.banderchuk@gmail.com}}%
\thanks{$^{2}$D. Coutinho and E. Camponogara are with the Department of Automation and Systems, Universidade Federal de Santa Catarina, Florianópolis, SC, 88040-900, Brazil,
        e-mail: {\tt\small daniel.coutinho(eduardo.camponogara)@ufsc.br}}%
\thanks{DOI: 10.1109/CDC49753.2023.10383966}
}

\begin{document}

\maketitle
\thispagestyle{empty}
\pagestyle{empty}

\begin{abstract}
This paper proposes a control strategy consisting of a robust controller and an Echo State Network (ESN) based control law for stabilizing a class of uncertain nonlinear discrete-time systems subject to persistent disturbances. Firstly, the robust controller is designed to ensure that the closed-loop system is Input-to-State Stable (ISS) with a guaranteed stability region regardless of the ESN control action and exogenous disturbances. Then, the ESN based controller is trained in order to mitigate the effects of disturbances on the system output. A numerical example demonstrates the potentials of the proposed control design method.
\end{abstract}

\section{INTRODUCTION}

{Artificial Intelligence tools have been used to learning and control  of dynamic phenomena that are difficult to model correctly such as  unknown, complex and non-linear systems.} For instance, \cite{ref:Li-Zhou-Luo-2018} {utilizes} a type of recurrent neural network {(RNN)} for redundant manipulators motion control in noisy environments,  {and} \cite{ref:Pan-Wang-2012} {employs} two {recurrent} neural networks for different tasks in a model predictive control approach for unknown nonlinear dynamical systems. {In particular, neural networks have also been used to replace controllers as in \cite{ref:Waegeman-Wyffels-Schrauwen-2012}, which considers a discrete-time RNN for learning a controller based on the inverse model approach applied to 
diverse linear and non-linear plants. 
However, RNNs are in general hard to be trained due to problems related to the back-propagation technique for long term time intervals; see, e.g., \cite{BSF94}.
%
%
Among several types of recurrent neural networks available in specialized literature, 
Echo State Networks (ESNs) seem to be more suitable for real-time implementations \cite{ref:Mahmoud-Elshenawy-2015}, since it employs}
fast linear regression algorithms for training \cite{ref:Jaeger-2001}. {Due to fast learning capabilities,  ESNs have attracted recurrent interest of control practitioners such as the works \cite{ref:Waegeman-Wyffels-Schrauwen-2012} and \cite{ref:Mahmoud-et-al-2021} which consider control structures based on ESNs for different classes of linear and non-linear plants.}

{A major drawback of applying RNNs for modeling and control of dynamics systems is the lack of formal verification tools providing safety and performance guarantees~\cite{BP02}. Recently,
several works have considered robust control theoretic tools for addressing stability issues related to closed-loop systems with the addition of learning feedback components. For instance, \cite{ref:Nguyen-et-al-2021} reformulates the feedback system in terms of a nonlinear differential inclusion in order to obtain stability certificates, \cite{ref:Soloperto-et-al-2018} proposed a learning-based robust model predictive control approach for linear systems with bounded state-dependent uncertainties, 
\cite{ref:Maiworm-et-al-2021} utilized a model predictive control in combination with an online learning of a non-linear Gaussian Process model scheme with guaranteed input-to-state stability, and \cite{ref:Bethge-et-al-2018} proposed a structure where robust feasibility and constraint satisfaction are guaranteed by nominal models, while performance is optimized using learned models.}

{Most of the very recent results applying robust control tools to derive stability and performance guarantees for feedback systems with learning capabilities are based on semi-definite programming (SDP) approaches; see, e.g., \cite{ref:Nguyen-et-al-2021}, \cite{PGBA21}, \cite{PKBKA22}, \cite{zenati2022validation} and
\cite{RWM21} to cite a few. However, when considering ESNs for modeling and control of complex systems, SDP tools 
may lead to a prohibitive computational effort,  since ESN dynamic models typically have a large number of states. In this paper, we propose a different solution to obtain stability guarantees for closed-loop systems with an embedded learning control law to avoid large computations when dealing with large-scale feedback systems. In particular, a two-loop strategy is considered (consisting of a robust controller with an ESN-based outer loop controller) in order to mitigate the effects of persistent disturbances for a class of (possibly open-loop unstable) uncertain nonlinear systems. The robust controller is firstly designed to guarantee that the closed-loop system is input-to-state stable (ISS) regardless of the ESN-based control law and exogenous disturbances, while minimizing the disturbance effects utilizing SDP tools. Then, the ESN-based control law is trained in order to improve the closed-loop performance thanks to the ISS property.   

The remainder of this paper is structured as follows. Section~\ref{sec:ps} describes the problem to be addressed in this paper and Section~\ref{sec:ir} introduces some basic results on ISS and ESNs, which are instrumental to derive the main result of this paper as established in Section~\ref{sec:cd}. Section~\ref{sec:ie} illustrates the application of the proposed results to a simulation example, whereas concluding remarks are drawn in Section~\ref{sec:cr}.}

\vskip 2mm

\noindent {\bf Notation:} $\bbz$ is the set of integers, $\bbz_{\geq}$ is the set of non-negative integers, $\bbr$ is the set of real numbers, $\bbr_{\geq}$ is the set of non-negative real numbers, $\bbr^{n \times m}$ is the set of $n \times m$ real matrices. For a vector sequence $x(k)$, $k=0,1,\dots,\infty$, the one-step ahead time-shift operation is denoted by $x_+ = x(k+1)$, where the argument $k$ of $x(k)$ is often omitted, and $\|x\|_\infty : = \sup_{k \in \bbz_\geq} \|x(k)\|$, where $\|x(k)\| : = \sqrt{x(k)^T x(k)}$. A function $\alpha : \bbr_{\geq} \rightarrow \bbr_{\geq}$ is a class $\mathcal{K}$-function if it is continuous, strictly increasing and $\alpha(0)=0$, and $\beta : \bbr_{\geq} \times \bbr_{\geq} \rightarrow \bbr_{\geq}$ is a class $\mathcal{K}\mathcal{L}$-function if $\beta(s,t)$ is a class $\mathcal{K}$-function on $s$ for each fixed $t \geq0$ and $\beta(s,t)$ is decreasing in $t$ and $\beta(s,t) \rightarrow 0$ as $t \rightarrow \infty$ for each fixed $s \geq0$. 

\section{PROBLEM STATEMENT} \label{sec:ps}

Consider the following class of nonlinear 
discrete-time systems:
\begin{equation} \label{eqn:sys}
\mathcal{G} : 
\left\{
\begin{aligned}
 x_+ & = f(x,\theta,u,d) \\
 y & = C x
\end{aligned}
\right.
\end{equation}
where $x \in \mcx \subset \bbr^{n_x}$ is the state vector, $\theta \in \Theta \subset \bbr^{n_\theta}$ is a vector of time invariant uncertain parameters, $u \in \bbr^{n_u}$ is the controlled input, $d \in \mcd \subset \bbr^{n_d}$ is the disturbance input, $y \in \bbr^{n_y}$ is the controlled output, $f(\cdot)$ is a polynomial vector function of $x$ and linear with respect to $(\theta,u,d)$, and $\mcx$, $\Theta$ and $\mcd$ are compact sets. It is assumed with respect to system~\eqref{eqn:sys} that:

\begin{enumerate}[{A}1)]

\item $f(0,\theta,0,0) = 0$ for all $\theta \in \Theta$.

\item $\mcx$ is a convex set, with $0 \in \mcx$, which can be represented in terms of either the convex hull of its $n_v$ vertices, \textit{i.e.}, 
\begin{equation} \notag
\mcx : = \mathrm{Co}\{v_1, \ldots, v_{n_v}\}, 
\end{equation}
or, alternatively, as the intersection of $n_h$ half-planes
\begin{equation} \notag
\mcx : = \{x \in \bbr^{n_x} : h_i^T x \leq 1, \ i = 1, \ldots, n_h\} 
\end{equation}
with $v_i \in \bbr^{n_x}$, $i=1,\ldots,n_v$, and 
$h_j \in \bbr^{n_x}$, $j~=~1,\ldots,n_h$, defining respectively the vertices and faces of~$\mcx$.

\item $\Theta$ is a polytopic set with known vertices. 

\item $\mcd$ is a magnitude bounded domain defined as
\begin{equation} \label{eqn:mcd}
   \mcd : = \{ d \in \bbr^{n_d} : d^T d \leq 1/(2\eta_d^{2}) \}
\end{equation}
with $\eta_d>0$ defining the size of $\mcd$.
\end{enumerate}

\vskip 2mm

In this paper, we are interested in regulating the controlled output $y$ (around a desired reference $r$) while mitigating the effects of (magnitude bounded) disturbance $d$ considering that the controlled input $u$ is generated by the combination of a robust control law $u_1$ and a correction term $u_2$ (determined by an Echo State Network -- ESN) as illustrated in Fig.~\ref{fig:cs}. 

\begin{figure}[thpb]
\begin{center}
\resizebox{\columnwidth}{!}{%
\begin{tikzpicture}[auto,>=latex]
    \node [input, name=input] {};
    \node [sum, right = of input] (sum) {};
    \node [block, right = of sum] (controller) {$\begin{array}{c} \mbox{Robust} \\ \mbox{Controller} \end{array}$};
    \node [sum, right = of controller] (sum2) {+};
    \node [block, right = of sum2, pin={[pinstyle]above:$d$}, 
            node distance=3cm] (system) {$\begin{array}{c} \mbox{Uncertain} \\ \mbox{Nonlinear System} \end{array}$};
    \node [block, above = of sum2] (controller2) {$\begin{array}{c} \mbox{ESN-based} \\ \mbox{Controller} \end{array}$};
    \draw [->] (sum2) -- node[name=u] {$u$} (system);
    \node [output, right =of system] (output) {};

    \draw [draw,->] (input) -- node[name=r] {$r$} (sum);
    \draw [draw,-] (r) -- ++ (0,1.7cm) -| (controller2.west);
    \draw [->] (sum) -- node {$e$} (controller);
    \draw [->] (controller) -- node {$u_1$} (sum2);
    \draw [->] (controller2) -- node {$u_2$} (sum2);
    \draw [->] (system) -- node [name=y] {$y$}(output);
    \draw [->] (y) -- ++ (0,-2cm) -| node[pos=0.99] {$-$} node [near end] {} (sum);
    \draw [-] (y) -- ++ (0,1.65cm) -| (controller2.east);
\end{tikzpicture}%
}
\caption{Proposed Control Setup.\label{fig:cs}}
\end{center}
\end{figure}
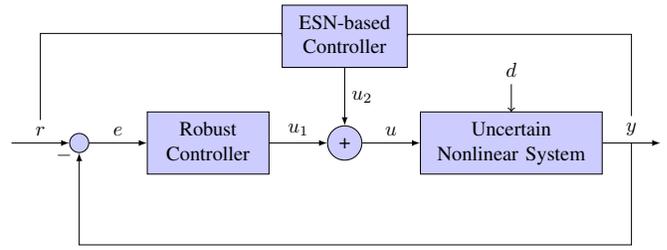

The robust control signal $u_1$ is designed to ensure that the closed-loop system is input-to-state stable regardless the persistent disturbance $d$ and the control signal $u_2$ (assuming that both are magnitude bounded with known limits). The correcting signal $u_2 \in \mcu$ is computed by means of an ESN which is designed to mitigate the effects of $d$ on $y$. It should be noted that this control scheme guarantees that the closed-loop system is ISS. 

For simplicity of presentation, this paper focus on the problem of regulating $y$ around $r\equiv0$ and it is assumed that:
\begin{enumerate}[{A}4)]

\item $u_2$ is constrained to the following set
\begin{equation} \label{eqn:mcu}
\mcu : = \{ u_2 \in \bbr^{n_u} : u_2^T u_2 \leq 1/(2 \eta_u^2) \},    
\end{equation}
with $\eta_u >0$ defining the size of $\mcu$.

\end{enumerate}

\vskip 2mm

In addition, we consider that the robust control signal is a polynomial state feedback control law as given below
\begin{equation} \label{eqn:u1}
   u_1 = K(x) x ,    
\end{equation}
with $K : \bbr^{n_x} \rightarrow \bbr^{n_u} \times \bbr^{n_x}$ being a polynomial matrix function to be determined. 

In view of the above scenario, the problems to be addressed in this paper are as follows:
\begin{itemize}

\item {\bf Robust Control Problem:} design $K(x)$ such that the closed-loop system of \eqref{eqn:sys} with $u=u_1+u_2$ is ISS for all $x(0) \in \mcr \subset \mcx$, $u_2 \in \mcu$ and $d \in \mcd$ while minimizing $\|x\|$, with $\mcr$ representing an over-bounding estimate of the reachable set.

\item {\bf ESN-based Control Problem:} for a given stabilizing robust control law $u_1$, train the ESN such that $\|y\|$ is minimized. 

\end{itemize}

\vskip 2mm

\begin{rmk}
The proposed methodology can be applied to systems whose origin is open-loop unstable, thanks to the robust controller which ensures that the state trajectory is bounded for given scalars $\eta_d$ and $\eta_u$. 
\end{rmk}

\begin{rmk}
We have adopted echo state networks due to their fast training, however alternative recurrent neural networks may be considered in the proposed control setup, possibly Gated Recurrent Units and Encoder-Decoder networks.
\end{rmk}

\vskip 2mm

\section{INSTRUMENTAL RESULTS} \label{sec:ir}

In this section, we recall some fundamental results regarding input-to-state stability and echo state networks which will be instrumental to derive the main result of this paper. 

\subsection{Input-to-State Stability}

Let the following discrete-time system
\begin{equation} \label{eqn:dtns}
    x_+ = f(x,w), \quad x(0) = x_0, \quad x \in \mcx , \quad w \in \mcw,
\end{equation}
where $x \in \mcx$, $w \in \mcw$, $\mcx$ is a compact set containing the origin and
\begin{equation} \label{eqn:mcw}
\mcw : = \{ w \in \mathbb{R}^{n_w} : w^T w \leq 1 \}    
\end{equation}

\begin{dfn}{(Input-to-State Stability~\cite{JW01,SCJ15})} \label{dfn:iss}
The origin of system \eqref{eqn:dtns} is said to be ISS if there exist a $\mathcal{K}$-function $\alpha : \bbr_{\geq} \rightarrow \bbr_{\geq}$, a $\mathcal{K}\mathcal{L}$-function $\beta : \bbr_{\geq} \times \bbr_{\geq} \rightarrow \bbr_{\geq 0 }$ and a positive scalar $\varrho$ such that the following holds:
\begin{equation}
\|x(k,x_0,w)\| \leq \alpha(\|w\|) + \beta(\|x_0\|,k)    
\end{equation}
for all $k \in \bbz_{\geq}$, $\|x_0\|\leq \varrho$ and $\|w\|_\infty \leq \varrho$.
\end{dfn}

\vskip 2mm

The following lemma is a Lyapunov characterization of input-to-state stability according to Definition~\ref{dfn:iss}, which has been adapted from~\cite[Lemma 1]{SCJ15} to our context.

\begin{lem}{(Lyapunov ISS Characterization)} \label{lem-iss-lyap}
The origin of system \eqref{eqn:dtns} is ISS if there exist $V : \mcx \rightarrow \bbr_{\geq}$ and a scalar $\mu \in (0,1)$ such that the following holds
\begin{equation} \label{eqn:lemma-iss-lyap-1}
\Delta V \leq \mu (w^T w - V) , \ \Delta V : = V_+ - V ,     
\end{equation}
for all $x \in \mcx$ and $w \in \mcw$. Moreover, for all $x(0) \in \mcr$ and $w \in \mcw$, the state trajectory $x(k) \in \mcr$ for all $k\geq0$, where 
\begin{equation} \label{eqn:mcr}
\mcr : = \{ x \in \mcx : V \leq 1 \}     
\end{equation}
satisfies $\mcr \subset \mcx$. 
\end{lem}

\vskip 2mm

\subsection{Echo State Networks}

Echo State Networks (ESNs) are a type of recurrent neural networks with fast learning that consist of an input layer, a recurrent layer with a large number of sparsely connected neurons (the reservoir), and an output layer~\cite{H02}. The connecting weights of the input and reservoir layers are fixed after initialization, and the output weights are easily trainable by means of linear regression problems~\cite{SSCZHL22}. Under some mild assumptions, the internal stability and the echo state property (\textit{i.e.}, the influence of initial conditions progressively vanishes with time) are guaranteed. 
    In other words, the echo state property refers to the ability of an ESN to maintain a stable internal state, regardless of the input it receives, allowing the network to continue processing new inputs without being affected by the previous inputs.

The ESN dynamics can be described in the following discrete-time state space representation:
\begin{equation} \label{eqn:ESN-ds}
\left\{
\begin{aligned}
   \xi_+ & = (1-\gamma) \xi + \gamma f_\xi\big( W_R^R \xi + W_\upsilon^R \upsilon + W_{\mathrm{bias}}^R\big) \\
   \sigma & = W_R^\sigma \xi
\end{aligned}
\right.
\end{equation}
where $\xi \in \bbr^n$ is the ESN state vector which corresponds to the reservoir neurons, $\upsilon \in \bbr^{n_\upsilon}$ is the ESN input, $\sigma \in \bbr^{n_\sigma}$ is the ESN output, $f_\xi : \bbr^n \times \bbr^{n_\upsilon} \times \bbr^n \rightarrow \bbr^n$ is the activation function -- typically element-wise $\tanh(\cdot)$, $\gamma \in (0,1)$ is the leak rate (a low pass filter constant), and $W_{R}^R$, $W_\upsilon^R$ and $W_R^\sigma$ are the reservoir-to-reservoir, input-to-reservoir and reservoir-to-output weight matrices, respectively, while $W_{\mathrm{bias}}^R$ represents a reservoir bias term.    

The connections going to the reservoir are randomly initialized and remain fixed. An ESN is typically  initialized by the following steps:
\begin{enumerate}[1)]
\begin{samepage} 
\item Matrices $\overline{W}_R^R \in \bbr^{n \times n}$, $W_\upsilon^R \in \bbr^{n \times n_\upsilon}$ and $W_{\mathrm{bias}}^R \in \bbr^n$ are randomly generated according to a normal distribution $\mathcal{N}(0,1)$.
\item The matrix ${W}_R^R$ is obtained by re-scaling $\overline{W}_R^R$ such that its spectral radius is smaller than $1$ (to ensure the internal stability). That is:
\begin{equation}
W_R^R = \frac{\rho_R}{\lambda_{\max}} \overline{W}_R^R,
\end{equation}
where $\rho_R \in (0,1)$ is the desired spectral radius and $\lambda_{\max}$ is the largest singular value of $\overline{W}_R^R$.
\item $W_{\upsilon}^R$ and $W_{\mathrm{bias}}^R$ are multiplied by scaling factors $\rho_{\upsilon}$ and $\rho_{\mathrm{bias}}$.
\end{samepage}
\end{enumerate}
These scaling parameters, $\rho_{R}$, $\rho_{\upsilon}$ and $\rho_{\mathrm{bias}}$ are key to the learning performance of the network, having an impact on the nonlinear representation and memory capacity of the reservoir. Also, low values for the leak rate $\gamma$ grant higher memory
capacity in reservoirs, while high values favor quickly varying inputs and outputs.

From a sufficiently large sequence of inputs and outputs (collected from the system), the reservoir-to-output weight $W_R^\sigma$ is trained by solving a least-squares problem.
    To train an ESN, the input data $\upsilon[k]$ is arranged in a matrix $\Upsilon$ and the desired output $\sigma[k]$ in a vector $\Sigma$ over a simulation time period, where each row $\upsilon^T$ of $\Upsilon$ corresponds to a sample time $k$ and its columns are related to the input units. 
For the sake of simplicity, we assume that there are multiple inputs and only one output. 
  The rows of $\Upsilon$ are input into the network reservoir according to each sample time, inducing a state matrix $\Xi$ with the resulting sequence of states in its rows.
According with Ridge Regression, the reservoir-to-output weight $W_R^\sigma$ is a vector obtained by solving the following linear system:
\begin{equation}
   (\Xi^T\Xi -\lambda I)W_R^\sigma = \Xi^T \Sigma  \label{eq:ridge-regres:ESN}
\end{equation}
where $\lambda$ is the Tikhonov regularization parameter, which serves to penalize the weight
magnitude, avoiding overfitting.
  In case of multiple outputs, the weight vector $W_R^{\sigma_i}$ for each output $\sigma_i$ is computed by solving the same equation \eqref{eq:ridge-regres:ESN}, however using the vector $\Sigma_i$ with the desired outputs. 

\vskip 2mm

\section{CONTROL DESIGN} \label{sec:cd}

The main results of this paper are introduced in this section. Firstly, a semi-definite programming approach is devised for designing a stabilizing state-feedback robust controller. Then, 
the ESN-based controller is designed considering the inverse dynamic control strategy. 

\subsection{Robust Control Design}

The system $\mathcal{G}$ in \eqref{eqn:sys} can be cast without loss of generality as follows
\begin{equation} \label{eqn:G2}
\mathcal{G} : \left\{
\begin{aligned}
 x_+ & = A(x,\theta) x + B_u(\theta)(u_1 + u_2) + B_d(\theta) d \\
 y & = C x
\end{aligned}
\right.    
\end{equation}
from the fact that $f(x,\theta,u,d)$ is polynomial w.r.t. (with respect to) $x$ and linear w.r.t. $(\theta,u,d)$, where $A(\cdot) \in \bbr^{n_x \times n_x}$ is a polynomial matrix function of $x$ and affine w.r.t. $\theta$, and $B_u(\cdot) \in \bbr^{n_x \times n_u}$ and $B_d(\cdot) \in \bbr^{n_x \times n_d}$ are affine matrix functions of $\theta$. 

\vskip 1mm

\begin{rmk}
    For simplicity of presentation, the input matrices $B_u$ and $B_d$ of system $\mathcal{G}$ are constrained to be state independent. However, we can easily handle state dependent matrices considering the following matrix decompositions
\begin{align*}
B_u(x,\theta) & = B_{u0}(\theta) + \Pi(x)^T B_{u1} \\
B_d(x,\theta) & = B_{d0}(\theta) + \Pi(x)^T B_{d1} 
\end{align*}    
with the matrices $B_{u0}$, $B_{u1}$, $B_{d0}$ and $B_{d1}$ being affine functions of $\theta$ having appropriate dimensions.     
\end{rmk}

\vskip 1mm

In order to derive a convex solution to the robust control problem, let $q$ be the largest degree of the polynomial entries of $A(x,\theta)$ and the following definitions:
\begin{align}
B_w(\theta) & = \begin{bmatrix} \eta_u B_u(\theta) & \eta_d B_d(\theta) \end{bmatrix}  , \ w = \begin{bmatrix} u_2/\eta_u \\ d/\eta_d  \end{bmatrix} , \notag \\
& \hspace{40mm} n_w = n_u + n_d \label{eqn:Bw} \\[1mm]
A(x,\theta) & = A_0(\theta) + \Pi(x)^T A_1(\theta) , \label{eqn:A0-A1}
\end{align}
where $A_0(\theta) \in \bbr^{n_x \times n_x}$ and $A_1(\theta) \in \bbr^{q n_x \times n_x}$ are affine matrix functions of $\theta$, and
\begin{equation} \label{eqn:Pi}
\Pi(x) = \begin{bmatrix}
    m^{(1)}(x) \otimes I_{n_x} & \cdots & m^{(q)}(x) \otimes I_{n_x}
\end{bmatrix}^T,    
\end{equation}
with $m^{(l)}(x) \in \bbr^{n_l}$, $l = 1,\ldots,q$, representing a vector which entries are all the monomials of degree $l$. 

Notice in view of \eqref{eqn:A0-A1} and \eqref{eqn:Pi} that there exist affine matrix functions $\Omega_0(x) \in \bbr^{n_m n_x \times n_x}$ and $\Omega_1(x) \in \bbr^{n_m n_x \times n_m n_x}$ such that the following hold:
\begin{align}
& \Omega_0(x) + \Omega_1(x) \Pi(x) = 0_{n_m n_x \times n_x} \label{eqn:Omega}\\ 
& \det\{\Omega_1(x)\} = c , \ c \neq 0 , \ \forall \ x \in \bbr^{n_x} \label{eqn:O1}
\end{align}
with $c$ being a constant real scalar and
\begin{equation} \label{eqn:nm}
n_m = n_1 + \cdots + n_q
\end{equation}

In addition, the control gain $K(x)$ of \eqref{eqn:u1} is defined as follows:
\begin{equation} \label{eqn:Kx}
K(x) = K_0 + K_1 \Pi(x)    
\end{equation}
where $\Pi(x)$ is as in \eqref{eqn:Pi}, and $K_0 \in \mathbb{R}^{n_u\times n_x}$ and $K_1 \in \bbr^{n_u \times n_mn_x}$ are matrices to be determined. 

Hence, taking \eqref{eqn:G2}, \eqref{eqn:Bw} and \eqref{eqn:A0-A1} into account, the closed-loop dynamics reads as follows
\begin{multline} \label{eqn:cls}
   x_+ \!=\! \bigg( A_0(\theta) \!+\! \Pi(x)^T A_1(\theta) \!+\! B_u(\theta) \big( K_0 \!+\!  K_1 \Pi(x) \big) \bigg) x \\
   + B_w(\theta) w 
\end{multline}

The following result proposes sufficient conditions for designing $K_0$ and $K_1$ such that the above system is ISS for all $w \in \mathcal{W}$ in terms of a finite set of LMI constraints. 

\vskip 1.5mm

\begin{thm}
Consider the system defined in \eqref{eqn:sys} satisfying (A1)-(A4), its closed-loop dynamics in \eqref{eqn:cls}, with $\Pi(x)$ satisfying \eqref{eqn:Omega} and \eqref{eqn:O1}. Let $\mu \in (0,1)$ be a given real scalar. Suppose there exist $Q=Q^T$, $G$, $M_0$, $M_1$ and $L$ be real matrices, with appropriate dimensions, satisfying the following LMIs:
\begin{equation}
    \begin{bmatrix} 1 & h_i^T Q \\
    Q h_i & Q \end{bmatrix} > 0 , \ i=1,\ldots,n_h \label{eqn:LMI1} \vspace{-2mm}
\end{equation}
\begin{multline}
\hspace{-4mm} 
    \begin{bmatrix}
        (1\!-\!\mu) \big( Q \!-\!G\! - \! G^T\big) & \star & \star & \star & \star \\
        0 & 0 & \star & \star & \star \\ 
        0 & 0 & -\mu I_{n_w} & \star & \star \\
        A_0(\theta)G \! + \! B_u(\theta) M_0 & B_u(\theta) M_1 & B_w(\theta) & - Q & \star \\
        A_1(\theta) G & 0 & 0 & 0 & 0 
    \end{bmatrix} \\
    + L \Omega(x) + \Omega(x)^T L^T < 0, \ \forall \ (x,\theta) \in \mathcal{V} \{ \mathcal{X} \times \Theta \} \label{eqn:LMI2} 
\end{multline}
where $\mathcal{V}\{ \mcx \times \Theta\}$ stands for the set of all vertices of $\mcx \times \Theta$ and
\begin{equation*}
    \Omega(x) = \begin{bmatrix}
\Omega_0(x) & \Omega_1(x) & 0 & 0 & 0 \\
0 & 0 & 0 & \Omega_0(x) & \Omega_1(x)
    \end{bmatrix} .
\end{equation*}
Then, the closed system in \eqref{eqn:cls} with
\begin{equation}
    K_0 \!=\! M_0 G^{-1} , \ K_1 \!=\! M_1 G_a^{-1} , \  
    G_a \!=\! \mathrm{diag}\{G, \ldots, G\} ,
\end{equation}
is locally ISS stable. Moreover, for all $x(0) \in \mcr$ and $w(k) \in \mcw$, the state trajectory $x(k)$ remains in $\mcr$, $\forall \ k > 0$, with
\begin{equation} \label{eqn:mcr2}
    \mcr  = \{ x \in \bbr^{n_x} : x^T Q^{-1} x \leq 1\} .
\end{equation}
\end{thm}

\begin{proof}
Firstly, note from \eqref{eqn:LMI1} that $Q\!>\!0$ and by applying the Schur's complement that $1 - h_i^T Q h_i  > 0$, for $i=1,\ldots,n_h$, which respectively imply that $V(x) = x^T Q^{-1} x > 0$, $\forall \ x\neq0$, and $\mcr \subset \mcx$, with $\mcr$ as defined in \eqref{eqn:mcr2}; see, e.g., \cite{BEFB}.

Next, considering the set of LMIs in \eqref{eqn:LMI2}, it follows that \eqref{eqn:LMI2} holds for all $(x,\theta) \in \mcx \times \Theta$ from convexity arguments. Hence, pre- and post-multiplying the matrix inequality in \eqref{eqn:LMI2} by $\Pi_a^T$ and $\Pi_a$, respectively, and from the fact that $\Omega(x) \Pi_a = 0$, yields 
\begin{equation} \label{eqn:proof1}
\Phi(x,\theta) \!<\! 0 , \ \forall \
(x,\theta) \in \mcx \times \Theta ,
\end{equation}
where
\begin{align} 
\Phi(x,\theta) & \!=\!    
\begin{bmatrix}
       (1-\mu)(Q-G-G^T) & \star & \star \\
       0 & -\mu I_{n_w} & 0 \\
       \left( \begin{array}{c} 
       \big(A_0(\theta)\!+\!\Pi^T A_1(\theta)\big)G \!+\! \\ B_u(\theta)\big(M_0\!+\!M_1 \Pi(x)\big) \end{array} \right) & B_w(\theta) & -Q
   \end{bmatrix}  \nonumber \\[1mm]
    \Pi_a & = \begin{bmatrix}
        I_{n_x} & 0 & 0 \\ \Pi(x) & 0 & 0 \\
        0 & I_{n_w} & 0 \\ 0 & 0 & I_{n_x} \\ 0 & 0 & \Pi(x) 
    \end{bmatrix} . \label{eqn:proof2}   
\end{align}

Notice from the block $(1,1)$ of $\Phi(x,\theta)$
that $Q < G + G^T$, which implies that $G$ is full rank since $Q> 0 $. Now, let the following parameterizations $M_0 = K_0 G$ and $M_1 = K_1 G_a$. Hence, by noting that $G_a \Pi(x) = \Pi(x) G$ and the fact that 
\begin{equation*}
    - G^T Q^{-1} G \leq Q - G - G^T ,
\end{equation*}
it holds from \eqref{eqn:proof1} that
\begin{equation} \label{eqn:proof3}
G_d^T \Psi(x,\theta) G_d < 0 , \ \forall \
(x,\theta) \in \mcx \times \Theta 
\end{equation}
with $G_d = \diag\{G,I_{n_w},I_{n_x}\}$ and
\begin{equation*}
\Psi(x,\theta) \!=\! 
\begin{bmatrix}
       (1-\mu)Q^{-1} & \star & \star \\
       0 & -\mu I_{n_w} & 0 \\
       \left( \begin{array}{c} 
       A_0(\theta)\!+\!\Pi^T A_1(\theta)\!+\! \\ B_u(\theta)\big(K_0\!+\!K_1 \Pi(x)\big) \end{array} \right) & B_w(\theta) & - Q
   \end{bmatrix}    
\end{equation*}

Then, by applying the Schur's complement to $\Psi(x,\theta)<0$ and by pre- and post-multiplying the resulting matrix inequality by 
$\begin{bmatrix} x^T & w^T \end{bmatrix}$ and its transpose, respectively, leads to
\begin{align}
    x_+^T Q^{-1} x_+ - (1-\mu) x^T Q^{-1} x - \mu w^T w < 0 \  \Rightarrow \nonumber \\
 \Delta V \leq \mu \big( w^T w - V\big) < 0 , \ \forall \
(x,\theta) \in \mcx \times \Theta, 
\end{align}
where $V = x^T Q^{-1} x$. The rest of this proof follows from Lemma~\ref{lem-iss-lyap}.
\end{proof}

\vskip 1mm

The robust control law $u_1 = K(x)x$ can be synthesized in order to minimize the effects of $w(k) \in \mcw$ on $x(k)$ by considering the following optimization problem
\begin{equation} \label{eqn:op}
\min_{\lambda,Q,\ldots,L,\mu} \ \lambda : 
        \begin{cases} 
        \lambda I_{n_x} - Q \geq 0, \\
        \mbox{ \eqref{eqn:LMI1} and \eqref{eqn:LMI2},}
        \end{cases}
\end{equation}
which minimizes the largest eigenvalue of $Q$. Notice that the above optimization problem implies that $\|x\|_\infty \leq \sqrt{\lambda}$.

\vskip 1mm


\vskip 1mm
\subsection{ESN Learning and Testing}

Several standard control algorithms use prior knowledge about a system to accomplish the desired behavior, such as the robust controller presented above. However, complex nonlinear systems may not be fully known or modeled correctly, prompting the application of a learning approach such as neural networks.
   Here, we consider a learning control based on the inverse model of the plant, which takes the desired output of a system and calculates the input that is needed to achieve that output. In other words, it ``inverts'' the relationship between inputs and outputs in the system. 
   
An ESN is proposed as the inverse model for its desirable features \cite{WwS12}. Being a type of recurrent neural network, the ESN has an internal memory that allows it to maintain a context or state across time steps, enabling it to better handle temporal dependencies. ESNs are also easy and fast to train.

  To learn the inverse model of the plant, this work considers the set-up shown in Fig.~\ref{fig:inverse-model:set-up}.
     The ESN input $\upsilon=(\tilde{y},\tilde{u}_1,\tilde{u}_2)$ collects past plant outputs $y[\cdot]$, and control signals $u_1[\cdot]$ and $u_2[\cdot]$, spaced in time according to a delay parameter $\delta\in\mathbb{N}$.
   Namely, $\tilde{y} = (y[k], y[k-\delta], ... , y[k-m\delta])$ consists of the current and past system outputs, $\tilde{u}_1 = (u_1[k-\delta], ..., u_1[k-m\delta])$ and $\tilde{u}_2 = (u_2[k-2\delta], ..., u_2[k-(m+1)\delta])$ collect $m$ previous control inputs.
Given the network input $\upsilon[k]$ at time $k$,  the goal of the inverse model is to learn the plant control input $u_2[k-\delta]$, at time $k-\delta$, that drove the system to the current output $y[k]$.
   The number of past outputs $m$ and the delay $\delta$ are hyper-parameters to be properly tuned.

    By simulating the plant with $u_1$ being the signal from the robust controller and $u_2$ being a random signal within the stability bounds, we obtain a matrix $\Xi$ with the trajectory of reservoir states and a vector $\Sigma$ with the desired outputs.  
Then ridge regression is applied to train the inverse model using Eq.~\eqref{eq:ridge-regres:ESN}, obtaining the reservoir-to-output matrix $W_{R}^{\sigma}$.

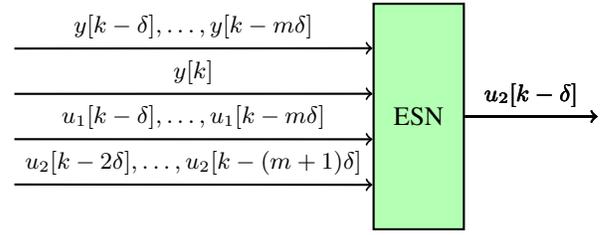
\begin{figure}[htbp]

    \tikzset{myblock/.style = {rectangle, draw}  }
  \begin{center}
  \begin{adjustbox}{max width=\textwidth} 
  \begin{tikzpicture}
    \node (in)[myblock,minimum height=3cm,draw=none] {};
    \node (esn)[myblock,right of=in,xshift=4.5cm,thick,
                minimum height=3cm,minimum width=1.2cm,fill=green!30]{ESN};
    \node (out)[myblock,right of=esn,xshift=1.5cm,minimum height=3cm,draw=none] {};
    \foreach \a/\b in {0.20/{$y[k-\delta],\dots,y[k-m\delta]$},
                       0.4/{$y[k]$},
                       0.6/{$u_1[k-\delta],\dots,u_1[k-m\delta]$},
                       0.8/{$u_2[k-2\delta],\dots,u_2[k-(m+1)\delta]$}}{
    \draw[->,thick]
              ($(in.north east)!\a!(in.south east)$) -- node[above] {\small \b} 
                  ($(esn.north west)!\a!(esn.south west)$);
    \draw[->,thick] (esn.east) -- node[above] {\small $u_2[k-\delta]$} (out.west);
   }
\end{tikzpicture}
\end{adjustbox} 
\end{center}

   \caption{General set-up for learning the inverse model.\label{fig:inverse-model:set-up}}
\end{figure}

Once the reservoir-to-output matrix $W_{R}^{\sigma}$ is learned for the inverse model, the resulting ESN can be used to generate the control signal $u_2$ that improves performance. This is achieved by {time shifting $\delta$ steps in} all signals of the learning set-up of Fig.~\ref{fig:inverse-model:set-up}, which leads to the set-up for control with the inverse model as shown in Fig.~\ref{fig:inverse-model:control}. At time $k$, by  setting $y[k+\delta]$ equal to the desired reference at time $k+\delta$ {(i.e., $r[k+\delta]$\footnote{Recall that in this paper, we are only considering the regulation problem around $y=0$ and thus $r[k+\delta]=0$.}),} the ESN will output the control signal $\overline{u}_2[k]$ that should be input at the current time to drive the plant to the desired reference.

\begin{figure}[htbp]
  
    \tikzset{myblock/.style = {rectangle, draw}  }
  \begin{center}
  \begin{tikzpicture}
    \node (in)[myblock,minimum height=3cm,draw=none] {};
    \node (esn)[myblock,right of=in,xshift=4cm,thick,
                minimum height=3cm,minimum width=1.2cm,fill=green!30]{ESN};
    \node (out)[myblock,right of=esn,xshift=1cm,minimum height=3cm,draw=none] {};
    \foreach \a/\b in {0.20/{$y[k],\dots,y[k-(m-1)\delta]$},
                       0.4/{$y[k+\delta]=r[k+\delta]$},
                       0.6/{$u_1[k],\dots,u_1[k-(m-1)\delta]$},
                       0.8/{$u_2[k-\delta],\dots,u_2[k-m\delta]$}}{
    \draw[->,thick]
              ($(in.north east)!\a!(in.south east)$) -- node[above] {\small \b} 
                  ($(esn.north west)!\a!(esn.south west)$);
    \draw[->,thick] (esn.east) -- node[above] {\small $\overline{u}_2[k]$} (out.west);
   }
\end{tikzpicture}
\end{center}
  \caption{General set-up for inverse model control.\label{fig:inverse-model:control}}
\end{figure}
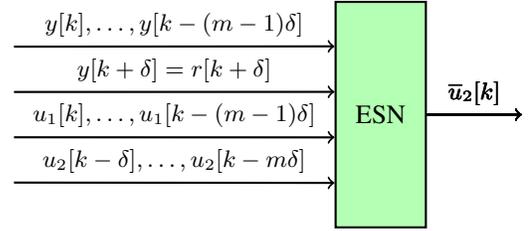
{Before feeding the ESN control law to the control system, $u_2$ is activated by the following function
\begin{equation*}
    u_2 = \frac{1}{\sqrt{2} \ \! \eta_u} \tanh{(\sqrt{2} \ \! \eta_u\overline{u}_2)}
\end{equation*}
ensuring that $u_2 \in \mcu$.
}

\section{ILLUSTRATIVE EXAMPLE} \label{sec:ie}


{Consider the following system which consists of the Van der Pol equation with the addition of an integrator:}
\begin{equation} \label{eqn:vdp-ct}
\left\{
\begin{aligned}
 \dot x_1(t) & \!= x_2(t) \\
 \dot x_2(t) & \!= - x_1(t) \!+\! \theta \big(1 \!-\! x_1(t)^2\big) x_2(t) \!+\! u(t) \!+\! d(t) \\
 \dot x_3(t) & \!= {x_1(t)} \\
 y(t) & \! = x_1(t)
\end{aligned}
\right.
\end{equation}
{where $x := \begin{bmatrix} x_1 & x_2 & x_3 \end{bmatrix}^T \in \mcx$ is the state vector, $\theta \in \Theta \subset \bbr$ is an uncertain parameter, $u \in \bbr$ is the control input, $d \in \mcd \subset \bbr$ is an exogenous disturbance, and $y \in \bbr$ is the controlled output.}  

{The control objective is to regulate $y(t)$ around zero considering the following sampled control law:
\begin{equation} \label{eqn:u-sh}
u(t) = u[kT_s] , \ \forall \ t \in [kT_s,(k\!+\!1)T_s), \ k \geq 0 ,
\end{equation}
where $T_s$ is a sufficiently small sampling period. In addition, the discrete-time control signal $u[k T_s] \!=\! u[k] \!=\! u$ will be determined utilizing the control strategy described in Section~\ref{sec:cd}. To this end, by applying the Euler's forward approximation, the quasi-linear discrete-time representation of \eqref{eqn:vdp-ct} as given in \eqref{eqn:G2} is defined by the following matrices:}
\begin{equation*}
\begin{aligned}
{A(x, \theta)} &= { 
\left[\begin{matrix} 1 & T_s  & 0 \\[1mm]
- T_s & 1 \!-\! T_s \theta \left(x_{1}^{2} \!-\! 1 \right)  & 0 \\[1mm]
T_s & 0 & 1
\end{matrix}\right]},  \\[1mm]
B_u(\theta) \!=\! B_u &= \left[\begin{matrix}0\\T_s\\0\end{matrix}\right] \; \text{ and } \; B_w(\theta) \!=\! B_w = \left[\begin{matrix}0 & 0\\T_s & T_s\\0 & 0\end{matrix}\right] .
\end{aligned}
\end{equation*}

{In this example, it will be assumed that:
\begin{itemize}

\item $\mcx = \{ x \in \bbr^3: |x_1| \leq 2 \}$; 

\item $\Theta = \{ \theta \in \bbr: 0.5 \leq \theta \leq 0.9 \}$;

\item $\mcd = \{ d \in\bbr: |d | \leq 1 / \sqrt{2} \}$; and 
\item $\mcu = \{ u_2 \in \bbr: | u_2 | \leq 1 / \sqrt{2} \}$.

\end{itemize}
Notice, in this particular example, that $\mcx$ is only bounded on the $x_1$ direction, since $A(x,\theta)$ is only a function of $x_1$ and $\theta$.}

{In order to design the robust controller, the system matrix is cast as follows 
\begin{equation*}
\left\{
\begin{aligned}
A(x,\theta) & = A_0(\theta) + \Pi(x)^T A_1(\theta) \\ 
0_{6 \times 3} & = \Omega_0(x) + \Omega_1(x) \Pi(x)  
\end{aligned} \right.
\end{equation*}
where}
\begin{equation*}
\begin{aligned}
&A_0(\theta) = \left[\begin{matrix}1 & T_{s} & 0\\- T_{s} & T_{s} \theta + 1 & 0\\T_{s} & 0 & 1\end{matrix}\right], A_1(\theta) = \left[\begin{matrix}0 & 0 & 0\\0 & 0 & 0\\0 & 0 & 0\\0 & 0 & 0\\0 & - T_{s} \theta & 0\\0 & 0 & 0\end{matrix}\right] \\
&\Pi(x) = \left[\begin{matrix}x_{1}\otimes I_3 \\ x_{1}^{2}\otimes I_3 \end{matrix}\right], \;\; \Omega_0(x) = \left[\begin{matrix} x_{1} I_3 \\ 0_3\end{matrix}\right] \;  \text{ and } \\
&\Omega_1(x) = \left[\begin{matrix} -I_3 & 0_3 \\ x_{1} I_3  & -I_3 \\\end{matrix}\right] .
\end{aligned}
\end{equation*}

{Hence, by considering $T_s = 0.1$ s {and $\theta=0.75$}, 
the optimization problem \eqref{eqn:op} is solved by applying a line search over $\mu \in (0,1)$ 
leading to the robust control law 
\begin{equation*}
    u = (K_0 + K_1\Pi(x))x = K(x)x
\end{equation*}
for an optimal $\mu = 0.3$, where}
\begin{equation*}
K(x)^T = \left[\begin{matrix}- 0.0217 x_{1}^{2} - 3.514 \cdot 10^{-14} x_{1} - 68.42\\0.7203 x_{1}^{2} + 1.058 \cdot 10^{-14} x_{1} - 16.73\\- 0.01905 x_{1}^{2} - 1.001 \cdot 10^{-13} x_{1} - 90.99\end{matrix}\right] ,
\end{equation*}
and the following reachable set estimate 
\begin{multline*}
\mathcal{R} = \{ x \in \bbr^3 : 2.2 x_{1}^{2} + 0.36 x_{1} x_{2} + 7.1 x_{1} x_{3} \\ + 0.03 x_{2}^{2} + 0.54 x_{2} x_{3} + 8.2 x_{3}^{2} \leq 10^{-3} \} .   
\end{multline*}

Then, for training the ESN-based controller, a data-set with 5000 samples was constructed based on simulations of the closed-loop system (i.e., with $u[k] = u_1[k]$) by considering that $w$ and $u_2$ were filtered white noises with cutoff frequencies defined to take into account typical low frequency disturbances. In addition, the following was assumed for the ESN hyper-parameters: $(i)$ a spectral radius $\rho = 0.5$; $(ii)$ reservoir size $N = 200$; $(iii)$ leaking rate $\gamma = 0.6$; and $(iv)$ reservoir density of $0.9$. These parameters were chosen considering a grid search and the neural network training error as the performance index. The number of past outputs ($m=1$) and the delay samples ($\delta=2$) were tuned based on the root mean square (RMS) 
value of the system output $y$ 
when adding the ESN control action (i.e., $u = u_1 + u_2$) relative to the robust control action only (i.e., $u = u_1$).

\vskip 2mm

The performance of the proposed control strategy is evaluated in the sequel by means of numerical simulations considering the plant continuous-time model and the sampled control law as defined in \eqref{eqn:vdp-ct} and \eqref{eqn:u-sh}, respectively. In particular, Fig.~\ref{fig:example:yw}
shows the disturbance attenuation properties of the proposed controller (i.e., $u = u_1+u_2$) compared to the robust control action only (i.e., $u = u_1$). 
%
The results show that the combination of robust control actions and ESN outperformed the performance obtained only with the robust controller by an improvement factor of $54.36\%$, which was calculated based on the RMS value of system output $y$ when adding the ESN control action to the robust control law relative to the robust controller only (which was the same metric utilized to determine $\delta$ and $m$ hyper-parameters for training the ESN-based controller). Furthermore, Fig.~\ref{fig:example:x1x2x3}  shows the phase portrait of the state trajectories (for both control laws) and the estimate $\mcr$ of the closed-loop reachable set. It is worth to mention that, as expected, the state trajectories remain confined to the set $\mathcal{ R}$ for all $t\geq0$.
\begin{figure}[htbp]
   \vspace{3mm}
   \begin{center}   
    \includegraphics[width=\columnwidth]{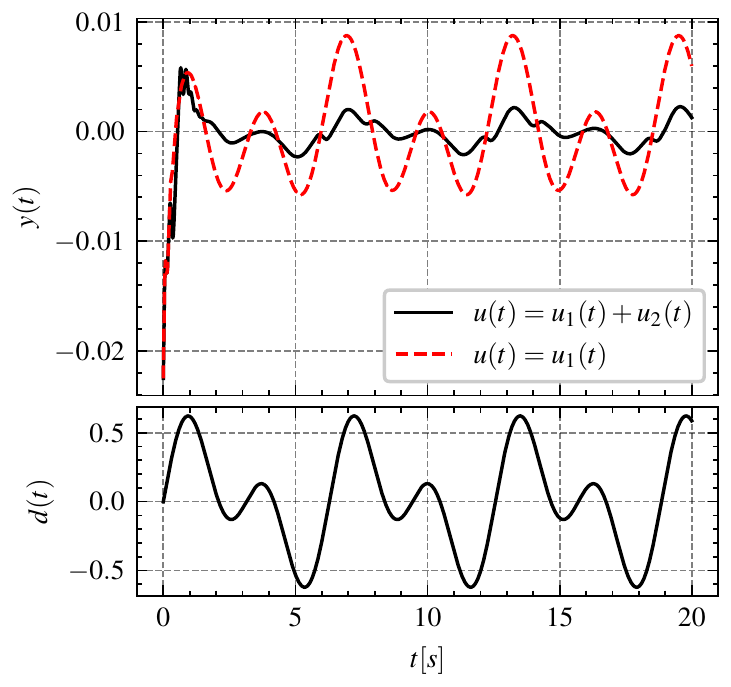}
   \end{center}
   \caption{System closed-loop response (at the top of figure) considering $u=u_1$ (red dashed line) and $u=u_1+u_2$ (black solid line) for $x_0^T = [ -0.0225 \;\;0.252 \;\; 0.005]$, {$\theta=0.75$} and the disturbance signal (depicted at the bottom of the figure) defined as $d(t) = 0.25\sqrt{2}(\sin(t) + \sin(2t))$. \label{fig:example:yw}}
\end{figure}
\begin{figure}[htbp]
   \begin{center}
     \includegraphics[]{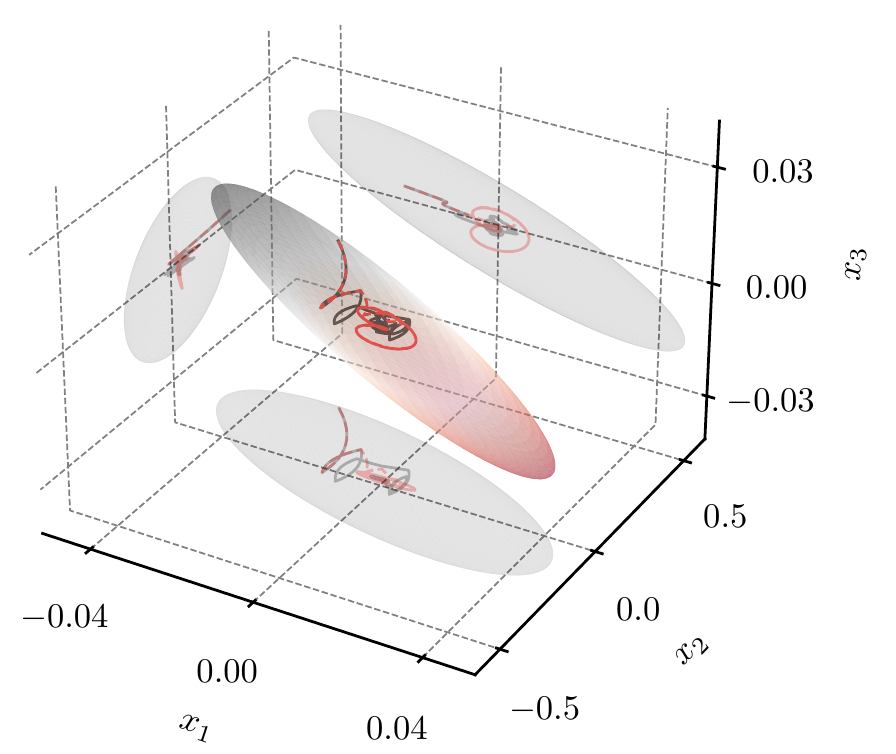}
   \end{center}
   \caption{Reachable set estimate and phase portrait of state trajectories for $x_0^T = [ -0.0225 \;\;0.252 \;\; 0.005]$,  {$\theta=0.75$} and $d(t) = 0.25\sqrt{2}(\sin(t) + \sin(2t))$, considering the control laws $u = u_1$ (state trajectory in red solid line) and $u = u_1 + u_2$ (state trajectory in black solid line).
   \label{fig:example:x1x2x3}}
\end{figure}


\section{CONCLUDING REMARKS} \label{sec:cr}

{This paper has proposed a robust control strategy with learning capabilities (based on ESNs) for stabilizing a class of uncertain polynomial discrete-time systems subject to unknown magnitude bounded disturbances. Firstly, a nonlinear state feedback is designed to ensure that the state trajectory driven by nonzero initial conditions and persistent disturbances is bounded to a positively invariant set regardless of the ESN control law (assuming a bounded action). Secondly, the ESN-based controller is trained to mitigate the effects of disturbances on the system output. A numerical example demonstrates the effectiveness of the proposed control technique. Future research will be concentrated on devising a robust controller with online learning capabilities.}

\bibliography{main}









\end{document}